\documentclass[reqno]{amsart}

\usepackage{amsmath,amsfonts}
\usepackage{epsfig} 
\usepackage{graphicx}

\def\cS{{\mathcal R}}

\def\e{\epsilon}

\def\tr{{\rm Tr}}

\def\N{{\mathbb N}}

\def\R{{\mathbb R}}

\def\ux{{\underline{x}}}

\def\frH{{\mathfrak H}}

\def\cH{{\mathcal H}}

\def\cM{{\mathcal M}}
\def\cN{{\mathcal N}}

\def\cR{{\mathcal R}}
\def\cS{{\mathcal S}}

\def\1{{\bf 1}}

\def\eqnn{\begin{eqnarray*}}
\def\eeqnn{\end{eqnarray*}}
\def\eqn{\begin{eqnarray}}
\def\eeqn{\end{eqnarray}}

\def\prf{\begin{proof}}
\def\endprf{\end{proof}}

\theoremstyle{plain}
\newtheorem{theorem}{Theorem}[section]

\newtheorem{lemma}[theorem]{Lemma}

\numberwithin{equation}{section}

\begin{document}

\parskip=8pt

\title[Well-posedness and scattering for the GP hierarchy]
{On the well-posedness and scattering for the Gross-Pitaevskii hierarchy via quantum de Finetti}

\author[T. Chen]{Thomas Chen}
\address{T. Chen,  
Department of Mathematics, University of Texas at Austin.}
\email{tc@math.utexas.edu}

\author[C. Hainzl]{Christian Hainzl}
\address{C. Hainzl,
Fachbereich Mathematik,  Universit\"at T\"ubingen, Germany.}
\email{christian.hainzl@uni-tuebingen.de}

\author[N. Pavlovi\'{c}]{Nata\v{s}a Pavlovi\'{c}}
\address{N. Pavlovi\'{c},  
Department of Mathematics, University of Texas at Austin.}
\email{natasa@math.utexas.edu}

\author[R. Seiringer]{Robert Seiringer}
\address{R. Seiringer,
Institute of Science and Technology Austria (IST Austria).}
\email{robert.seiringer@ist.ac.at}

\maketitle

\begin{abstract}
We prove the existence of scattering states for the defocusing cubic 
Gross-Pitaevskii (GP) hierarchy in $\R^3$.
Moreover, we show that an exponential energy growth condition commonly 
used in the well-posedness theory
of the GP hierarchy is, in a specific sense, necessary.
In fact, we prove that without the latter, there exist initial data for the 
focusing cubic GP hierarchy for which instantaneous blowup occurs.
\end{abstract}

\section{Introduction} 
\label{sec-mainres-1}

The cubic Gross-Pitaevskii (GP) hierarchy is a system of infinitely many coupled linear PDE's
describing a Bose gas of infinitely many
particles, interacting via two-body delta interactions (repulsive in the defocusing case, and 
attractive  in the focusing case). It emerges in the derivation of the nonlinear Schr\"odinger equation (NLS) from
a bosonic $N$-particle Schr\"odinger system in the limit  as $N\rightarrow\infty$, 
where the pair interaction potentials 
tend to a delta distribution.
In this paper, we prove the existence of scattering states for solutions to the defocusing cubic 
GP hierarchy in $\R^3$. Moreover, we show that an exponential energy growth condition commonly 
used in the well-posedness theory
of the GP hierarchy is, in a specific sense, necessary.

Our approach uses the {\em quantum de Finetti theorem} as presented in the work
of Lewin, Nam and Rougerie \cite{lnr} (see Section \ref{ssec-deFinetti-1}).  
We previously applied it in \cite{chpa} to give a new, short proof of the unconditional uniqueness of solutions
to the cubic GP hierarchy in the energy space. The quantum de Finetti theorem allows
us to lift a variety of results that hold for the corresponding NLS to the GP hierarchy. 
In particular, we use this approach in the work at hand
to establish the existence of scattering states for the cubic defocusing
GP hierarchy in $\R^3$.

Another main goal of this paper is to illuminate an important 
exponential energy growth condition that is 
invoked in all works on the well-posedness
of the GP hierarchy equations in the literature. We show that if this condition is 
removed, the focusing GP hierarchy equations
become ill-posed. Again, the de Finetti theorem allows us to access this
previously elusive problem by relating it to the blowup in $H^1$  of solutions to the corresponding
focusing cubic NLS.


The first derivation of the nonlinear Hartree (NLH) equation from an interacting Bose gas
was given by Hepp in \cite{he}, via second quantization and coherent states.  
Lanford, in his fundamental analysis of the $N\rightarrow\infty$ limit of $N$-particle
systems in classical mechanics, made central use of the BBGKY hierarchy \cite{Lan-1,Lan-2}.
The latter was subsequently employed by Spohn for a different derivation of the NLH, 
in \cite{sp}. Fr\"ohlich, Tsai and Yau revisited this topic more recently in \cite{frtsya}.
Subsequently, Erd\"os, Schlein and Yau gave the derivation of 
the NLS and NLH for a wide range of situations in their landmark works \cite{esy1,esy2,esy3,esy4}. 
In their approach, proving the {\em uniqueness of solutions} to the GP hierarchy in a space of marginal density matrices
$L^\infty_{t\in[0,T)}\frH^1$ (defined in \eqref{frakH-def-1} below) is a crucial ingredient.   
Their approach involves sophisticated singular integral estimates organized with 
Feynman graph expansions, and introduces an important combinatorial 
method that controls the large number of such graphs.

Subsequently, by combining a reformulation of the combinatorial method of \cite{esy1,esy2,esy3,esy4}
with methods from the theory of dispersive PDE's, 
Klainerman and Machedon \cite{KM}  gave a shorter proof of uniqueness of solutions
in a different solution space,
but under the assumption of an a priori condition on the solutions.  
Their approach was used by various authors for the derivation of the NLS 
from interacting Bose gases 
\cite{chpa,CPBBGKY,xch3,CheHol-2013,kiscst,CT1,zxie}.
The analysis of the Cauchy problem for the GP hierarchy was initiated in \cite{chpa2} and continued e.g. in \cite{GreSohSta-2012,CT1}. 

In \cite{CHPS-1}, we gave a new proof of unconditional uniqueness for solutions to the cubic GP hierarchy 
in $\R^3$. Our result is equivalent to the uniqueness result in \cite{esy1,esy2,esy3,esy4}; the proof
combines the Erd\"os-Schlein-Yau combinatorial method \cite{esy1,esy2,esy3,esy4} 
in boardgame formulation \cite{KM}, with an application of the {\em quantum de Finetti theorem} \cite{lnr},
see Section \ref{ssec-deFinetti-1}.

There exists a variety of different 
approaches to the derivation of the NLS and NLH
from many-body quantum dynamics, due to the contributions of many authors; we refer to 
\cite{esy1,esy2,esy3,esy4,ey,kiscst,rosc} and the references therein,
and also \cite{adgote,anasig,frgrsc,frknpi,frknsc,grma,grmama,he,pick,pick2}.
These dispersive nonlinear PDE's give a mean field description of the dynamics of Bose-Einstein condensates, 
whose formation was first experimentally verified in 1995,
\cite{anenmawico,dameandrdukuke}. For the mathematical study of Bose-Einstein condensation, 
we refer to \cite{ailisesoyn,lise,lisesoyn,liseyn, liseyn2} and the references therein.

\subsection{Definition of the GP hierarchy}

The cubic defocusing GP hierarchy on $\R^3$  for an infinite sequence of 
bosonic
marginal density matrices $\Gamma=(\gamma^{(k)})_{k\in\N}$ is defined as the initial value problem
\eqn \label{eq-def-GP}
	i\partial_t \gamma^{(k)} =\sum_{j=1}^k [-\Delta_{x_j},\gamma^{(k)}]   
	\, + \,   \lambda B_{k+1} \gamma^{(k+1)}
	\nonumber\\
	\gamma^{(k)}(0)=\gamma_0^{(k)} \;,
	\;\;\; k\in\N \;,
\eeqn
where $\lambda\in\{1,-1\}$, and
where $\gamma^{(k)}(t;\ux_k;\ux_k')$ is fully symmetric under permutations separately of the
components of $\ux_k:=(x_1,\dots,x_k)$, and of the components of $\ux_k':=(x_1',\dots,x_k')$. 
We call \eqref{eq-def-GP} {\em defocusing} if $\lambda=1$, and {\em focusing} if $\lambda=-1$.
The interaction term for the $k$-particle marginal is defined by
\eqn \label{eq-def-b}
	B_{k+1}\gamma^{(k+1)}
	\, = \, B^+_{k+1}\gamma^{(k+1)}
        - B^-_{k+1}\gamma^{(k+1)} \, ,
\eeqn
where 
\eqn\label{eq-Bplus-GP-def-1}
    B^+_{k+1}\gamma^{(k+1)}
     = \sum_{j=1}^k B^+_{j;k+1 }\gamma^{(k+1)},
\eeqn
and 
\eqn 
    B^-_{k+1}\gamma^{(k+1)}
    = \sum_{j=1}^k B^-_{j;k+1 }\gamma^{(k+1)},
\eeqn                  
with 
\begin{align}\label{eq-Bplusmin-def-1-1}
    & \left(B^+_{j;k+1}\gamma^{(k+1)}\right)(t,x_1,\dots,x_k;x_1',\dots,x_k') \nonumber\\
    & \quad \quad = \int dx_{k+1}  dx_{k+1}'  \nonumber\\
    & \quad\quad\quad\quad 
	\delta(x_j-x_{k+1})\delta(x_j-x_{k+1}' )
        \gamma^{(k+1)}(t,x_1,\dots,x_{k+1};x_1',\dots,x_{k+1}') \nonumber\\
    & \quad \quad =  
        \gamma^{(k+1)}(t,x_1,\dots,x_j,\dots,x_k,x_j;x_1',\dots,x_k',x_j),
\end{align} 
and 
\begin{align}\label{eq-Bplusmin-def-1-2}
    & \left(B^-_{j;k+1}\gamma^{(k+1)}\right)(t,x_1,\dots,x_k;x_1',\dots,x_k') \nonumber\\
    & \quad \quad = \int dx_{k+1} dx_{k+1}'  \nonumber\\
    & \quad\quad\quad\quad 
	  \delta(x'_j-x_{k+1})\delta(x'_j-x_{k+1}' )
        \gamma^{(k+1)}(t,x_1,\dots,x_{k+1};x_1',\dots,x_{k+1}')\nonumber\\
    & \quad \quad =  
        \gamma^{(k+1)}(t,x_1,\dots,x_k,x_j';x_1',\dots,x_j',\dots,x_k',x_j') \,.
\end{align} 
We say that $B^+_{j;k+1}$ {\em contracts} the triple of variables $x_j,x_{k+1},x_{k+1}'$, and that
$B^-_{j;k+1}$ contracts the triple of variables $x_j',x_{k+1},x_{k+1}'$.

In \cite{esy1,esy2,esy3,esy4} and \cite{CHPS-1},  the 
well-posedness of \eqref{eq-def-GP}  is studied in the space of solutions
\begin{align}\label{frakH-def-1}
 	\frH^1:=\Big\{ \,(\gamma^{(k)})_{k\in\N} \, \Big| \, 
 	\tr(| S^{(k,1)}  [\gamma^{(k)}] |) < R^{2k} \; \mbox{for some constant }R<\infty \, \Big\}
\end{align}
where $S^{(k,\alpha)}:=\prod_{j=1}^k(1-\Delta_{x_j})^{\alpha/2}(1-\Delta_{x_j'})^{\alpha/2}$ for $\alpha>0$.

We write
\eqn
     U^{(k)}(t) := \prod_{\ell=1}^k e^{it(\Delta_{x_\ell}-\Delta_{x_\ell'})}
\eeqn
for the free $k$-particle propagator.
A {\em mild solution} to \eqref{eq-def-GP} in the space $L^\infty_{t\in[0,T]}\frH^1$ 
is a sequence of marginal density matrices $\Gamma=(\gamma^{(k)}(t))_{k\in\N}$ 
solving the integral equation
\eqn
    \gamma^{(k)}(t) = U^{(k)}(t)\gamma^{(k)}(0) + i \int_0^t  U^{(k)}(t-s) B_{k+1}\gamma^{(k+1)}(s) ds \,
    \;\;\;,\;\;\; k\in\N\,,
\eeqn
satisfying  
\eqn
    \sup_{t\in[0,T]}\tr(|S^{(k,1)}[\gamma^{(k)}(t)]|) < R^{2k}
\eeqn
for a finite constant $R$ independent of $k$. 
 
\subsection{The cubic NLS} 

In the special case of factorized initial data,
\eqn\label{eq-gammak-factorized-1}
	\gamma^{(k)}_0(\ux_k;\ux_k') \, = \, \prod_{j=1}^k \phi_0(x_j) \, \overline{\phi_0(x_j')} \,,
\eeqn
the condition that $(\gamma_0^{(k)})\in\frH^1$ implies
\eqn 
    \tr(|S^{(k,1)}[\gamma_0^{(k)}] |) = \|\phi_0\|_{H^1}^{2k} < R^{2k} 
    \;\;\;,\;\;\;k\in\N\,,
\eeqn
and is equivalent to the condition $\|\phi_0\|_{H^1} < R$.
A particular solution to  \eqref{eq-def-GP} with initial data \eqref{eq-gammak-factorized-1} is given by
$\Gamma=(\gamma^{(k)}(t))_{k\in\N}$ where for all $k\in\N$,
\eqn
	\gamma^{(k)}(t;\ux_k;\ux_k') \,  =  \,
	\prod_{j=1}^{k} \phi(t,x_j) \, \overline{ {\phi}(t,x'_j) } 
\eeqn
is factorized.  
In particular, the 1-particle wave function $\phi$ satisfies the  
cubic NLS 
\eqn\label{eq-NLS-def-1}
    i\partial_t\phi(t)=-\Delta\phi(t)+ \lambda|\phi(t)|^2\phi(t) \;\;\;,\;\;\;\phi(0)=\phi_0\in H^1\,,
\eeqn 
which is defocusing if $\lambda=1$ and focusing if $\lambda=-1$.

Solutions to \eqref{eq-NLS-def-1} conserve the $L^2$-mass 
\eqn
    M[\phi(t)] =\|\phi(t)\|_{L^2_x}^2 = M[\phi_0] \,,
\eeqn 
the momentum
\eqn 
    P[\phi(t)] = i \int \overline{\phi(t,x)}\nabla\phi(t,x) dx \,,
\eeqn
angular momentum
\eqn 
    L[\phi(t)] = i\int \overline{\phi(t,x)}\,x\wedge\nabla\phi(t,x) dx \,,
\eeqn
and the energy
\eqn
    E[\phi(t)] = \frac12\| \nabla_x \phi(t) \|_{L^2_x}^2 + \frac\lambda4\|\phi(t)\|_{L^4_x}^4 = E[\phi_0] \,.
\eeqn
The cubic NLS in $\R^3$ \eqref{eq-NLS-def-1} is $L^2$-supercritical and $H^1$-subcritical,
and is globally well-posed in $H^1$ if $\lambda=1$,  
and locally well-posed if $\lambda=-1$, \cite{TaoBook}. 

\subsubsection{The defocusing NLS}

In the defocusing case $\lambda=1$, \eqref{eq-NLS-def-1} is globally well-posed and 
displays the existence of scattering states
and asymptotic completeness:

\begin{theorem}\label{thm-NLSscattering-1}
Let $S_t:\phi_0\mapsto \phi(t)$ denote the flow map associated to \eqref{eq-NLS-def-1}, 
for $t\in\R$ and $\lambda=1$.
Then, there exist  continuous bijections (wave operators) 
$W_+,W_-:H^1(\R^3)\rightarrow H^1(\R^3)$,  such that
the strong limit 
\eqn    
    \lim_{t\rightarrow\pm\infty}e^{-it\Delta}S_t(\phi_0) = \phi_\pm
    \;\;\;\; , \;\;\;\; \phi_0=W_\pm(\phi_\pm)
\eeqn
holds for all $\phi_0\in H^1(\R^3)$.
\end{theorem} 

We refer  to Section 3.6 in \cite{TaoBook} for a detailed discussion and a proof.

\subsection{The quantum de Finetti theorem}
\label{ssec-deFinetti-1}

As shown in our recent work \cite{CHPS-1},
solutions to the GP hierarchy and solutions to the NLS are closely interconnected via the 
{\em quantum de Finetti theorem}, which 
is a quantum analogue of the Hewitt-Savage theorem in probability theory, \cite{HewittSavage}.
We quote it in the formulation presented by Lewin, Nam and Rougerie in \cite{lnr} who coined the 
notions of the strong and weak quantum de Finetti theorems (here collected into a single theorem).

\begin{theorem}\label{thm-strongDeFinetti-1}
Let $\cH$ be a separable Hilbert space and let 
$\cH^k = \bigotimes_{sym}^k\cH$ denote the corresponding bosonic $k$-particle space. 
Let $\Gamma$ denote a collection of   
bosonic density matrices on  $\cH$, i.e.,
$$
\Gamma = (\gamma^{(1)},\gamma^{(2)},\dots)
$$
with $\gamma^{(k)}$ a non-negative trace class operator on $\cH^k$.
Then, the following hold:
\begin{itemize}
\item
(Strong Quantum de Finetti theorem, \cite{HudsonMoody,Stormer-69,lnr})
Assume that $\Gamma$ is admissible, i.e., $\gamma^{(k)}=\tr_{k+1} \gamma^{(k+1)}$, 
where $\tr_{k+1}$ denotes the partial trace over the $(k+1)$-th factor, $\forall k\in\N$. 
Then, there exists a unique Borel probability measure $\mu$, 
supported on the unit sphere in $\cH$, and invariant under multiplication of 
$\phi \in \cH$ by complex numbers of modulus one, such that 
\begin{equation}\label{gkdf}
    \gamma^{(k)} = \int d\mu(\phi)  (  | \phi  \rangle \langle \phi |  )^{\otimes k}
    \;\;\;,\;\;\;\forall k\in\N\,.
\end{equation} 
\item 
(Weak Quantum de Finetti theorem, \cite{lnr,AmmariNier-2008,AmmariNier-2011}) 
Assume that $\gamma_N^{(N)}$ is an arbitrary sequence of mixed states on $\cH^N$, $N\in\N$,
satisfying $\gamma_N^{(N)}\geq 0$ and $\tr_{\cH^N}(\gamma_N^{(N)})=1$, and assume
that its $k$-particle marginals have weak-* limits 
\eqn 
    \gamma^{(k)}_{N}:=\tr_{k+1,\cdots,N}(\gamma^{(N)}_N)
    \; \rightharpoonup^* \; \gamma^{(k)} \;\;\;\; (N\rightarrow\infty)\,,
\eeqn
in the trace class on $\cH^k$ for all $k\geq1$ (here, $\tr_{k+1,\cdots,N}(\gamma^{(N)}_N)$ 
denotes the partial trace in the $(k+1)$-st up to $N$-th component). 
Then, there exists a unique Borel probability measure $\mu$ on the unit ball in
$\cH$, and invariant under multiplication of 
$\phi \in \cH$ by complex numbers of modulus one, 
such that
\eqref{gkdf} holds
for all $k\geq0$. 
\end{itemize}
\end{theorem}

We note that the limiting hierarchies of marginal density matrices obtained via weak-* limits from
the BBGKY hierarchy of bosonic $N$-body 
Schr\"odinger systems as in \cite{esy1,esy2,esy3,esy4} do not necessarily satisfy admissibility.

For the problems considered in this paper, the Hilbert space is given by $\cH=L^2(\R^3)$.  
In  \cite{CHPS-1}, we have used Theorem \ref{thm-strongDeFinetti-1} to present a new, shorter proof of the
{\em unconditional}  uniqueness of solutions to the GP hierarchy
in $L^\infty_{t\in[0,T)}\frH^1$; we thereby also obtain a direct correspondence between solutions
to the GP hierarchy and solutions to the NLS which will be crucial for our proof of the main results in this paper.
The unconditional uniqueness part itself is equivalent to the uniqueness result proven in
\cite{esy1,esy2,esy3,esy4}.
Our main result in  \cite{CHPS-1} states the following.

\begin{theorem}\label{thm-uniqueness-1}  (Chen-Hainzl-Pavlovi\'c-Seiringer, \cite{CHPS-1})
Let $(\gamma^{(k)}(t))_{k\in\N}$ be a mild solution  in $L^\infty_{t\in[0,T)}\frH^1$ 
to the (de)focusing cubic GP hierarchy in $\R^3$ with initial data $(\gamma^{(k)}(0))_{k\in\N}\in\frH^1$, 
which is either admissible, or obtained
at each $t$ from a weak-* limit as described in  Theorem \ref{thm-strongDeFinetti-1}.  

Then, $(\gamma^{(k)})_{k\in\N}$ is the unique solution for the given initial data.

Moreover, assume that the initial data $(\gamma^{(k)}(0))_{k\in\N} \in\frH^1$ satisfy
\eqn 
    \gamma^{(k)}(0) = \int d\mu(\phi)(|\phi\rangle\langle\phi|)^{\otimes k} 
    \;\;\;,\;\;\;\forall k\in\N\,,
\eeqn   
where $\mu$ is a Borel probability measure
supported  either  on the unit sphere or on the unit ball in $L^2(\R^3)$,
and invariant under multiplication of 
$\phi \in \cH$ by complex numbers of modulus one.  
Then, 
\eqn\label{eq-GPdeF-NLS-sol-1}
    \gamma^{(k)}(t) = \int d\mu(\phi)(|S_t(\phi)\rangle\langle S_t(\phi)|)^{\otimes k} 
    \;\;\;,\;\;\;\forall k\in\N\,,
\eeqn
where $S_t:\phi\mapsto \phi(t)$ is the flow map of the cubic (de)focusing NLS, for $t\in[0,T)$. That is, $\phi(t)$
satisfies \eqref{eq-NLS-def-1} with initial data $\phi$.
Accordingly,
\eqn\label{eq-GPdeF-pullb-sol-1}
    \gamma^{(k)}(t) = \int d\mu_t(\phi)(|\phi\rangle\langle\phi|)^{\otimes k} 
    \;\;\;,\;\;\;\forall k\in\N\,,
\eeqn
where $d\mu_t(\phi):=d\mu(S_{-t}(\phi))$ is the push-forward measure under the NLS flow.
\end{theorem}

\section{Statement of main resuilts}

In this paper, we prove the existence of scattering
states for the defocusing cubic GP hierarchy in $\R^3$. 
Moreover, we investigate the necessity of the energy growth condition in the definition of the
solution spaces $\frH^1$, see \eqref{frakH-def-1}.

\subsection{Scattering for the cubic GP hierarchy in $\R^3$}

We prove the existence of scattering states using the quantum de Finetti theorems,
Theorem \ref{thm-NLSscattering-1},
and Theorem \ref{thm-uniqueness-1}, which was proved in our earlier paper \cite{CHPS-1}.
The initial data for the GP hierarchy $\Gamma_0=(\gamma_0^{(k)})_{k\in\N}$ have the form
\eqn\label{eq-gamma0-def-1}
    \gamma_0^{(k)} = \int  d\mu(\phi)\big( |\phi\rangle\langle\phi|\big)^{\otimes k} \,.
\eeqn 
We consider the defocusing cubic NLS with $\lambda=1$, and assume that  
\eqn\label{eq-muEphi-cond-1}
    \int d\mu(\phi) (E[\phi])^{2k} \leq R^k
\eeqn
holds for some finite constant $R>0$, and all $k\in\N$, where 
\eqn 
    E[\phi] = \frac12\int |\nabla\phi|^2dx +\frac14 \int |\phi|^4 dx \,,
\eeqn
is the energy functional for the cubic defocusing NLS in $\R^3$.
The condition \eqref{eq-muEphi-cond-1} is equivalent to $\mu$ having support in  a
ball in $H^1$; see Lemma \ref{lm-Chebyshev-1}, below.

We note that while the de Finetti theorems provide the existence and uniqueness of a 
measure $\mu$, $\mu$ is in general not explicitly known. 
Therefore, it is important to express the
condition \eqref{eq-muEphi-cond-1},  directly at the level of density matrices.
This is addressed in Section \ref{ssec-higheren-1} below, where we review {\em higher order
energy functionals} for GP hierarchies that were first introduced in \cite{CPHE}.

The first main result of this paper establishes the existence of scattering states for the 
cubic defocusing GP hierarchy on $\R^3$, and provides the construction
of the corresponding asymptotic measures for the de Finetti representation \eqref{eq-GPdeF-pullb-sol-1}. 
This has been a longstanding open problem
despite much activity in the field.
With our approach via the de Finetti theorem, it follows from the scattering theory for the NLS.

\begin{theorem}\label{thm-main-1}
Let $\Gamma_0=(\gamma_0^{(k)})_{k\in\N}$ be as in \eqref{eq-gamma0-def-1},
and $\lambda=1$ (the defocusing case).
We assume that $\mu$ satisfies 
\eqref{eq-muEphi-cond-1}.

Let $\gamma^{(k)}(t)=\int d\mu(\phi)(|S_t\phi\rangle\langle S_t\phi|)^{\otimes k}$, for $k\in\N$,
denote the unique solution to \eqref{eq-def-GP} 
satisfying $\gamma^{(k)}(0)=\gamma_0^{(k)}$, for $k\in\N$. 

Then, there exist unique asymptotic measures $\mu_+$, $\mu_-$ such that
\eqn
    \gamma_\pm^{(k)}:=\int d\mu_\pm(\phi) (|\phi\rangle\langle\phi|)^k  
\eeqn
are scattering states  $\gamma_+^{(k)}$, $\gamma_-^{(k)}$ 
on $L^2(\R^{3k})$ satisfying
\eqn
    \lim_{t\rightarrow\pm\infty}
    \tr\Big( \, \Big| \, S^{(k,1)} \Big[ \, U^{(k)}(-t)\gamma^{(k)}(t) 
    - \gamma_\pm^{(k)} \, \Big] \, \Big| \, \Big)
    =0
\eeqn
for all $k\in\N$. 
In particular,
\eqn 
    d\mu_\pm(\phi) = d\mu(W_\pm(\phi)) \,
\eeqn
where the continuous bijections $W_+$, $W_-:H^1\rightarrow H^1$ are the wave operators
from Theorem \ref{thm-NLSscattering-1}.
\end{theorem}

More generally, our method allows to transfer knowledge about the non-linear 
Schr\"odinger equation (as given in Theorem \ref{thm-NLSscattering-1}) to 
results about the GP hierarchy. For instance, if the existence of scattering states for the 
focusing NLS can be shown for a suitable set of initial data (see for instance \cite{DuyHolRou}), one can also infer a 
corresponding result for the GP hierarchy for initial states with de Finetti measure 
$\mu$ supported on that set.

\subsubsection{Higher order energy functionals}
\label{ssec-higheren-1}

The condition on $\mu$ given in \eqref{eq-muEphi-cond-1} 
can be formulated directly at the level of marginal density matrices.
This is of importance because the initial data for the GP hierarchy is
usually provided at the level of density matrices $\gamma_0^{(k)}$,
without explicit determination of the measure $\mu$.
To this end, we recall the higher order energy functionals that were introduced in \cite{CPHE}.
In the case of the cubic GP hierarchy, they are defined by
\begin{align}\label{eq-HE-def-1}
    \langle K^{(m)}\rangle_{\Gamma(t)}:=\tr(K^{(m)}\gamma^{(2m )}(t))
\end{align}
for $m\in\mathbb{N}$, where
\eqn
    K_\ell&:=&\frac{1}{2}(1-\Delta_{x_\ell})\text{Tr}_{\ell+1}+\frac{1}{4}B^+_{\ell;\ell+1} 
    \;\;\;,\;\;\;\ell\in\N\,, 
    \nonumber\\
    K^{(m)}&:=&K_1K_{3}\cdots K_{2m- 1}.
\eeqn
In \cite{CPHE}, it is shown that these higher order energy functionals are conserved.
 
We note that 
\eqn\label{eq-mu-condfoc-2-0}
    \tr(K_1K_3\cdots K_{2k-1} \gamma_0^{(2k)}) =
    \int d\mu(\phi) (E[\phi])^k  
\eeqn
corresponding to \eqref{eq-muEphi-cond-1}; see Section 4 of \cite{CPHE}.

\subsection{Energy growth condition}
 
Results on the well-posedness of the Cauchy problem for the GP hierarchy are usually
obtained in solution spaces of marginal density matrices where an exponential
growth condition either of the form 
\eqn\label{eq-H1growth-1}
    \tr|S^{(k,1)}[\gamma^{(k)}]| < R^{2k} \;\;\;\forall k\in\N
\eeqn
holds in the trace norm, or of the form
\eqn\label{eq-H1growth-2}
    \|S^{(k,1)}[\gamma^{(k)}] \|_{\rm HS} < R^{2k} \;\;\;\forall k\in\N
\eeqn
in the Hilbert-Schmidt norm.
In the works \cite{chpa,chpa2,CPHE,CPBBGKY,chpatz1,chpatz2}
and \cite{xch3,CheHol-2013,kiscst,CT1,zxie}, well-posedness is studied
in solution spaces incorporating the condition \eqref{eq-H1growth-2}.
In \cite{esy1,esy2} and the paper at hand, only the case \eqref{eq-H1growth-1} is considered;
a condition of this form  is an important technical ingredient
for these uniqueness proofs.
We would like to address the crucial question whether the 
energy growth condition \eqref{eq-H1growth-1} in the definition of the space $\frH^1$ 
is necessary for a well-posedness theory.

We introduce the quantity
\eqn\label{eq-logbound-1}
    \cR_{H^1}(\mu):= \exp\Big[ \, \limsup_{k\rightarrow\infty}\frac1{2k}\log\Big(\    
    \int d\mu(\phi) \|\phi\|_{H^1}^{2k} \, \Big)  \, \Big] \,,
\eeqn 
which corresponds to the radius of the smallest ball in $H^1$ that contains the support of $\mu$.
We observe that \eqref{eq-H1growth-1}, expressed 
via the de Finetti theorem as 
\eqn
    \int d\mu(\phi)\|\phi\|_{H^1}^{2k}<R^{2k} \;\;,\;\;\;\forall k\in\N\,,
\eeqn
is equivalent to
the condition that $\mu$ satisfies
\eqn 
    \cR_{H^1}(\mu) \, < \,  R \,,
\eeqn
for $R<\infty$. 
Hence,  \eqref{eq-H1growth-1} simply means that $\mu$ has bounded support in $H^1$.


Here, we prove that if a faster than exponential growth rate is
admitted, so that $\cR_{H^1}(\mu)=\infty$, 
the {\em  focusing} cubic GP hierarchy is {\em ill-posed}, in the sense that there exist
initial data at $t=0$ for which the solution blows up instantaneously;
that is, the norm
$\tr(|S^{(k,1)}\gamma^{(k)}(t)|)$ diverges for any positive $t>0$.

This result is a consequence of the following well-known 
result about the blowup in $H^1$ of solutions of the cubic NLS in the focusing case $\lambda=-1$.
Eq. \eqref{eq-NLS-def-1} is locally well-posed; given any initial data $\phi_0\in H^1$, 
there exists $\tau=\tau(\phi_0)>0$
and a unique solution $\phi(t)\in  H^1$ for $t\in[0,\tau)$.
However, the solution might only exist for a finite time.
Let 
\eqn 
    V[\phi](t) := \|x\phi(t)\|_{L^2}^2
\eeqn    
denote the quadratic moment in $x$ with respect to $\phi(t)$.
Then, blowup in finite time occurs whenever $E[\phi_0]<0$ and $V[\phi_0]<\infty$.
This is proven by use of the {\em virial identities} 
(Vlasov-Petrishchev-Talanov \cite{vlpetal}, Zakharov \cite{zakh}, Glassey \cite{glassey})
\eqn\label{eq-virial-1}
    \partial_t V[\phi](t) = 2 \Im  \int x\cdot \overline{\phi(x)}\nabla\phi(x) dx
\eeqn
and
\eqn\label{eq-virial-2}
    \partial_t^2 V[\phi](t) = 16 E[\phi_0] - 2 \|\phi(t)\|_{L^4}^4 \,.
\eeqn
In fact, if $E[\phi_0]<0$, the r.h.s. of \eqref{eq-virial-2}
is strictly negative, and therefore, $\|x\phi(t)\|_{L^2}$ tends to
zero in finite time.
However, by the Heisenberg uncertainty principle,
\eqn 
    \|\phi_0\|_{L^2}^2 
    \leq C \|x\phi(t)\|_{L^2}\|\phi(t)\|_{H^1} \,,
\eeqn
see for instance \cite{TaoBook}.
Hence, a bound of the form $\|x\phi(t)\|_{L^2}< b(t)$ with $b(t) \searrow0$ as $t\nearrow T= T(\phi_0)$ 
implies that the solution blows up in $H^1$, that is, 
$\|\phi(t)\|_{H^1}\nearrow\infty$ as $t\nearrow \tau=\tau(\phi_0)$ for some $\tau\leq T$.
We refer to $\tau(\phi_0)$ as the blowup time corresponding to the initial data $\phi_0\in H^1$.

One can easily derive an upper bound on the blowup time as follows. From 
the virial identity \eqref{eq-virial-1}, it follows that 
\eqn 
    |\partial_t \|x\phi(t)\|_{L^2}^2 | \leq 2 \|x\phi(t)\|_{L^2} \|\phi(t)\|_{\dot H^1} \,,
\eeqn
and from \eqref{eq-virial-2} that
\eqn 
    \partial_t^2 \|x\phi(t)\|_{L^2}^2 < 16 E[\phi(t)] = 16 E[\phi] \,,
\eeqn
where $\phi(t)$ solves the focusing cubic NLS with initial data $\phi$.
From second order Taylor expansion in $t$, we thus find that  
\eqn\label{eq-blowup-quadeq-1}
    \|x\phi(t)\|_{L^2}^2 &\leq& \|x\phi\|_{L^2}^2 +
     2 t \|x\phi\|_{L^2} \|\phi\|_{\dot H^1} +8t^2  E[\phi] \,.
\eeqn
While the left hand side is non-negative, the right hand side becomes negative in
finite time if $E[\phi]<0$, which implies that the solution blows up in $H^1$.
If $E[\phi]<0$, it follows that the quadratic equation on the right hand side has
precisely one positive and one negative root.
The positive root $T(\phi)>0$  is an upper bound
on the blowup time $\tau(\phi)$.

Combining this with the de Finetti representation \eqref{eq-GPdeF-NLS-sol-1}
for solutions to the GP hierarchy, we obtain the following main result.

\begin{theorem}\label{thm-main-2}
Consider the set of probability measures $\mu$  on the unit ball in $L^2(\R^3)$.
Then, the following dichotomy holds for the focusing cubic GP hierarchy \eqref{eq-def-GP} 
(where we have $\lambda=-1$):
\begin{itemize}
\item
For the subset of probability measures satisfying 
\eqn\label{eq-logbound-0}
    \cR_{H^1}(\mu) < \infty\,,
\eeqn 
the following holds.
Given $\mu_0\in\{\mu\,|\,\cR_{H^1}(\mu) <\infty\}$,  there exists a unique solution to 
the focusing cubic GP hierarchy in $L^\infty_{[0,T)}\frH^1$, 
for some $T=T(\mu_0)>0$, with the initial data 
\eqn\label{eq-gamma0-deF-mainthm2}
    \Big(\,\gamma^{(k)}_0 = \int d\mu_0(\phi) (|\phi\rangle\langle\phi|)^{\otimes k} 
    \, \Big)_{k\in\N} 
\eeqn 
in $\frH^{1}$.
\item
For the subset of probability measures satisfying 
\eqn\label{eq-logbound-1-0}
    \cR_{H^1}(\mu) = \infty\,,
\eeqn 
the following holds.
For any $\delta>0$, there exist probability measures $\mu_0\in \{\mu\,|\,\cR_{H^1}(\mu) =\infty\}$
with the following properties:
\begin{itemize} 
\item
The right hand side of \eqref{eq-logbound-1} diverges at a rate at most $\exp(c k^\delta)$ 
as $k\rightarrow\infty$,
\eqn\label{eq-logbound-2}
   \exp\Big[ \, \frac1{2k}\log\Big(\    
    \int d\mu_0(\phi) \|\phi\|_{H^1}^{2k} \, \Big)  \, \Big] \, < \, C e^{c k^\delta} \,.
\eeqn 
\item  
The initial data defined by $\mu_0$ as in 
\eqref{eq-gamma0-deF-mainthm2}
satisfies $\tr(|S^{(k,1)}\gamma_0^{(k)}|)<\infty$ for all $k\in\N$,
but the associated solution to the cubic focusing GP hierarchy
displays instantaneous blowup (see below for the precise definition). 
\end{itemize}
\end{itemize}
\end{theorem}

\subsection{Remarks}
We make the following remarks concerning the case \eqref{eq-logbound-1-0}: 
\begin{itemize}
\item
The precise meaning of instantaneous blowup that we are considering is as follows.
Let $A_{R}:=\{\phi\in L^2 | \,\|\phi\|_{L^2}=1 \,,\, \|\phi\|_{H^1}\leq R \}$ for $R>0$, and 
denote by $\1_{A_R}$ the corresponding characteristic function. Then, 
for every   $R>0$, there exists $T=T(R)>0$ such that the sequence of regularized density matrices
\eqn 
    \Big( \, \gamma^{(k)}_R(t) := \int d\mu_0(\phi) \1_{A_R}(\phi)  
    (|S_t(\phi)\rangle\langle S_t(\phi)|)^{\otimes k} \Big)_{k\in\N} 
\eeqn 
is a solution to the focusing cubic GP hierarchy in $L^\infty_{t\in[0,T(R))}\frH^{1}$. 
However, in the limit $R\rightarrow\infty$,
\eqn\label{eq-limlog-infty-1}
     \lim_{R\rightarrow\infty}\tr(|S^{(k,1)}[\gamma_R^{(k)}(t)] |) = \infty \;\;\;\forall t>0 \,,
\eeqn
for any $k\in\N$.
It is in this sense that we say that $\tr(|S^{(k,1)}[\gamma^{(k)}(t)] |)$ blows up instantaneously for $t>0$.  
\\
\item
We note that for local well-posedness to hold, it is necessary that
$\mu_0$-almost surely, the blowup time,
$\tau(\phi)>\e>0$, is bounded away from zero. 
In our analysis of the case \eqref{eq-logbound-1-0}, 
we will construct measures $\mu_0$ for which $\tau(\phi)$ can be arbitrarily small
on the support of $\mu_0$. This is only possible when $\|\phi\|_{H^1}$ can be arbitrarily large 
on the support of $\mu_0$.
\end{itemize}

\section{Proof of Theorem \ref{thm-main-1}}

In this section, we apply the quantum de Finetti theorem to prove 
the existence of scattering states for solutions 
to the defocusing cubic GP hierarchy in 3 dimensions. 

To begin with, we observe that the condition \eqref{eq-muEphi-cond-1}
implies that $E[\phi]\leq  R$ holds $\mu$-almost surely.

\begin{lemma}\label{lm-Chebyshev-1}
Assume that 
\eqn\label{eq-EnCond-1}
    \int d\mu(\phi) (E[\phi])^{2k} \leq R^{2k}
\eeqn
holds for some finite constant $R>0$, and all $k\in\N$. Then,  
\eqn
    \mu\Big( \Big\{ \, \phi\in L^2(\R^3) \, \Big| \, \|\phi\|_{L^2}=1 \;,\; E[\phi] >  R \,\Big\} \Big) \, = \, 0 \,.
\eeqn
\end{lemma}

\begin{proof}
From Chebyshev's inequality, we have that
\eqn
    \lefteqn{
    \mu\Big( \, \Big\{ \, \phi\in L^2(\R^3) \, \Big| \, \|\phi\|_{L^2}=1 \;,\; E[\phi] > \lambda \,\Big\} \, \Big)
    }
    \nonumber\\
    &\leq& \frac1{\lambda^{2k}}\int d\mu(\phi) (E[\phi])^{2k} \leq \frac{R^{2k}}{\lambda^{2k}} \,,
\eeqn
and for $\lambda> R$, the right hand side tends to zero when $k\rightarrow\infty$.
\end{proof}

Recalling that $\lambda=1$,
the representation \eqref{eq-gamma0-def-1} immediately yields
\eqn\label{eq-defoc-H1bd-1}
    \tr(|S^{(k,1)}[\gamma_0^{(k)}] |) &=& \int d\mu(\phi)\|\phi\|_{H^1}^{2k} 
    \nonumber\\
    &\leq&
     \int d\mu(\phi) (1+2E[\phi])^{k} \leq (1+2 R)^{k} \;\;\; \forall k\in \N \,.
\eeqn
This implies that $\mu$-almost surely, $\|\phi\|_{H^1}^2\leq 1+2E[\phi] \leq 1+2R$, by the
same argument as in Lemma \ref{lm-Chebyshev-1}.
Thus, Theorem \ref{thm-NLSscattering-1} implies that $\mu$-almost surely, there exists a unique
solution to the defocusing cubic NLS \eqref{eq-NLS-def-1} with initial data $\phi(0)=\phi$ which exhibits scattering and 
asymptotic completeness. For notational convenience
further below, we denote $g_\pm(\phi):=\phi_\pm$, such that
\eqn 
    \lim_{t\rightarrow\pm\infty} \|e^{-it\Delta}S_t(\phi)-g_\pm(\phi)\|_{H^1} = 0 \,.
\eeqn
Then, $g_\pm(\phi)=W_\pm^{-1}(\phi)$.

Using the de Finetti representation of
the $k$-particle marginal  
\eqn
    \gamma^{(k)} = \int  d\mu(\phi)\big( |\phi\rangle\langle \phi |\big)^{\otimes k} \,,
\eeqn
we let
\eqn\label{eq-gpm-def-1}
     \gamma_\pm^{(k)} 
     &:=&  \int  d\mu(\phi)\big( |g_\pm(\phi)\rangle\langle g_\pm(\phi)|\big)^{\otimes k} 
     \nonumber\\
     &=&\int  d\mu_\pm(\phi)\big( |\phi \rangle\langle \phi |\big)^{\otimes k}  \,,
\eeqn
where $d\mu_\pm(\phi)=d\mu(W_\pm(\phi))$.

It follows from energy conservation and positivity of the
potential energy term $\lambda\|\phi\|_{L^4}^4$ that $\mu$-almost surely,
\eqn\label{eq-phitg-bd-1}
    \|S_t(\phi)\|_{H^1}^2 &\leq& 1+2E[\phi] \; \leq \; 1+2R
    \nonumber\\
    \|g_+(\phi)\|_{H^1}^2 &\leq& 1+2 E[\phi] \; \leq \; 1+2R \,.
\eeqn
For $\phi\in H^1$ satisfying $E[\phi]<R$, we have
\eqn\label{eq-scattdiff-bd-1}
     \|e^{-it\Delta}S_t(\phi)-g_\pm(\phi)\|_{H^1}&\leq&  \|e^{-it\Delta}S_t(\phi)\|_{H^1}
     +\|g_\pm(\phi)\|_{H^1} 
     \nonumber\\
     &\leq& 2(1+2E[\phi])^{1/2} < 2(1+2R)^{1/2}
\eeqn
uniformly in $\phi$, and uniformly in $t\in\R$. 
Thus, we obtain that
\eqn\label{eq-domconv-1}
    \lefteqn{
    \lim_{t\rightarrow\pm\infty}
    \int d\mu(\phi) \|e^{-it\Delta}S_t(\phi)
    -g_\pm(\phi)\|_{H^1} 
    }
    \nonumber\\
    &=&
    \int d\mu(\phi) \lim_{t\rightarrow\pm\infty}\|e^{-it\Delta}S_t(\phi)
    -g_\pm(\phi)\|_{H^1}
    \nonumber\\
    &=&
    0 \, ,
\eeqn  
from the dominated convergence theorem.

We may now prove the existence of scattering states at the level of the GP hierarchy.
Using Theorem \ref{thm-main-1} and \eqref{eq-gpm-def-1}, we obtain that
\eqn
    \lefteqn{
    \tr\Big(\Big|S^{(k,1)}\Big[U^{(k)}(-t)\gamma^{(k)}(t) - \gamma_+^{(k)}) \Big]\Big|\Big)
    }
    \nonumber\\
    &=&\int d\mu(\phi) \tr\Big(\Big|S^{(k,1)}
    \Big[ \, 
    \big( \, |U(-t)S_t(\phi)\rangle\langle U(-t)S_t(\phi)| \, \big)^{\otimes k}
    \nonumber\\
    &&\hspace{5cm}
    - \big( |g_+(\phi)\rangle\langle g_+(\phi)|\big)^{\otimes k}  \Big) \Big]\Big|\Big) 
    \label{eq-auxdiff-1-1}
\eeqn
Using the identity 
\eqn  
    A_0^{\otimes k}-A_1^{\otimes k} 
    &=&
    \sum_{j=0}^{k-1} 
    A_{1}^{\otimes j}\otimes (A_0-A_1) \otimes A_{0}^{\otimes k-1-j}
\eeqn
with $A_0:=|U(-t)S_t(\phi)\rangle\langle U(-t)S_t(\phi)|$ and $A_1:=|g_+(\phi)\rangle\langle g_+(\phi)|$, 
and 
\eqn
    \tr(|S^{(1,1)}[A_0-A_1]|) \leq \|e^{-it\Delta}S_t(\phi) - g_+(\phi)\|_{H^1} \big( \, \|S_t(\phi)\|_{H^1} +
    \|g_+(\phi)\|_{H^1} \, \big)
\eeqn
we find 
\eqn 
    \eqref{eq-auxdiff-1-1}
    &\leq&\sum_{j=0}^{k-1}     
    \int d\mu(\phi)\tr(|S^{(1,1)}[A_0-A_1]|)(\tr(|S^{(1,1)}[A_1]|))^j\tr(|S^{(1,1)}[A_0]|)^{k-j-1}
    \nonumber\\
    &\leq&
    \int d\mu(\phi)\|e^{-it\Delta}S_t(\phi) - g_+(\phi)\|_{H^1} ( \, \|S_t(\phi)\|_{H^1} +
    \|g_+(\phi)\|_{H^1} \, )^{2k-1}
    \nonumber\\
    &\leq&\sum_{j=0}^{k-1}
    \Big(\int d\mu(\phi)\|e^{-it\Delta}S_t(\phi) - g_+(\phi)\|_{H^1}^{2k}\Big)^{\frac1{2k}} 
    \nonumber\\
    &&
    \hspace{1cm}
    \Big(\int d\mu(\phi)( \, \|S_t(\phi)\|_{H^1}+\|g_+(\phi) \|_{H^1} \, )^{2k}\Big)^{\frac{2k-1}{2k}} \,.
    \label{eq-scattdiff-1}
\eeqn
It follows from \eqref{eq-defoc-H1bd-1} that $\mu$-almost
surely, $E[\phi(t)]=E[\phi]<R$. 
Together with \eqref{eq-phitg-bd-1}, this implies
\eqn 
    \eqref{eq-scattdiff-1} \leq 
    2^k
    \Big(\int d\mu(\phi)\|e^{-it\Delta}S_t(\phi) - g_+(\phi)\|_{H^1}^{2k}\Big)^{\frac 1{2k}} 
    (1+2R)^{\frac{2k-1}{2k}} \,.
\eeqn
The right hand side converges to zero as $t\rightarrow\infty$, as a consequence of 
\eqref{eq-scattdiff-bd-1} and
\eqref{eq-domconv-1}.

This concludes the proof of Theorem \ref{thm-main-1}.
\qed

\section{Proof of Theorem \ref{thm-main-2}}

\subsection{The case $\cR_{H^1}(\mu)<\infty$}
Given $\cR_{H^1}(\mu)<R$ for some $R<\infty$, 
it follows from Lemma \ref{lm-Chebyshev-1} that $\mu$-almost surely,
$\|\phi\|_{H^1}<R$. 

The focusing cubic  NLS, with flow map $\phi\mapsto S_t(\phi)$, 
is locally well-posed in $H^1(\R^3)$. In particular, there exist constants  $T>0$ and $M<\infty$
such that
$\|S_t(\phi)\|_{H^1}<M$ for $t\in[0,T]$ where $T=T(\|\phi\|_{H^1})$ 
is monotonically decreasing, and where $M=M(\|\phi\|_{H^1})<\infty$ is monotonically 
increasing in $\|\phi\|_{H^1}$ (more details are given in Section \ref{ssec-TMmonot-1} below).
Thus, by monotonicity of $T$ and $M$ with respect to $ \|\phi\|_{H^1}$, 
it follows that $\mu$-almost surely,
$\|S_t(\phi)\|_{H^1}<M(R)<\infty$ for $t\in[0,T(R)]$.

Therefore, $\gamma^{(k)}(t)$ as given in \eqref{eq-GPdeF-NLS-sol-1}, with $k\in\N$,
satisfy 
\eqn
    \sup_{t\in[0,T(R)]}\tr(|S^{(k,1)}[\gamma^{(k)}]|)<(M(R))^{2k} \;,\;\;\forall k\in\N\,,
\eeqn 
and hence, 
$(\gamma^{(k)})_{k\in\N}\in L^\infty_{t\in[0,T(R)]}\frH^1$. This proves the existence of a solution, 
and its uniqueness follows from Theorem \ref{thm-uniqueness-1}.

\subsubsection{Monotonicity of the constants $T$ and $M$ with respect to $\|\phi\|_{H^1}$}
\label{ssec-TMmonot-1}
We remark that one can take $T(\|\phi\|_{H^1})\sim \|\phi\|_{H^1}^{-\beta}$ for some $\beta>0$ and 
$M(\|\phi\|_{H^1})\sim  \|\phi\|_{H^1}$. 
For example, this can be easily obtained from applying the estimate
(3.42) in \cite{chpa2} to factorized solutions to the GP hierarchy 
$\Gamma(t)=((|S_t(\phi)\rangle\langle S_t(\phi)|)^{\otimes k})_{k\in\N}$ with initial data of the focusing
cubic NLS satisfying
$\|\phi\|_{H^1}<R$, and for parameters $\xi_1=\frac{1}{2R}$ and $\xi_2=\frac{1}{4R}$.
In this case, we note that $\|\Gamma(t)\|_{\cH_{\xi_2}^1}=\sum_{k\geq1}\xi_2^k\|S_t(\phi)\|_{H^1}^{2k}$,
etc,
in the notation of \cite{chpa2}.

\subsection{The case $\cR_{H^1}(\mu)=\infty$}
We will explicitly construct a family of probability measures on the unit sphere 
in $L^2(\R^3)$ satisfying  
\eqn\label{eq-limlog-infty-1}
    \cR_{H^1}(\mu)=\infty \,
\eeqn 
with a prescribed maximum rate of divergence, together with
\eqn\label{eq-mumom-cond-1}
    \int d\mu(\phi)\|x\phi\|_{L^2}^2<\infty \,,
\eeqn 
and $\tr(|S^{(k,1)}\gamma_0^{(k)}|)<\infty$ for all $k\in\N$,
such that instantaneous blowup occurs for the corresponding initial data.

In fact, we will be more specific, and construct measures $\mu$ such that the sequence
$(\gamma^{(k)} = \int  d\mu(\phi)\big( |\phi\rangle\langle \phi |\big)^{\otimes k})_{k\in\N}$ belongs to the set
\begin{align}\label{eq-frakHeps-def-1} 
 	\frH^{\alpha,r}:=\Big\{ \,(\gamma^{(k)})_{k\in\N} \, \Big| \, 
 	\tr(| S^{(k,\alpha)}  [\gamma^{(k)}] |) < e^{ck^{r}} \; \mbox{for some constant }c<\infty 
 	\,\Big\} \,
\end{align} 
for $r\geq1$, where evidently, 
$\frH^{\alpha}=\frH^{\alpha,1}$. 

Instead of an exponential growth of order 
$\int d\mu(\phi)\|\phi\|_{H^1}^{2k}\leq R^k = O( e^{ck})$, our aim is to
admit a growth of order $O(e^{ck^r})$ for some arbitrary $r>1$.
We note that any probability measure $\mu$  on $L^2(\R^3)$ having the property that
\eqn\label{eq-mugamma-Hrcond-1}
    \Big( \, \int d\mu(\phi) (|\phi\rangle\langle\phi|)^{\otimes k} 
    \, \Big)_{k\in\N} \,
    \in \frH^{1,r}\setminus\frH^{1} \,
\eeqn 
satisfies \eqref{eq-limlog-infty-1}.  The parameter $r>1$ determines the rate of divergence
of \eqref{eq-logbound-1}.

To construct a measure $\mu$ satisfying \eqref{eq-mumom-cond-1}
and  \eqref{eq-mugamma-Hrcond-1}, we may, for simplicity,  pick $\mu$ to be supported on the unit sphere 
\eqn 
    \cS:=\{\psi\in L^2(\R^3)| \, \|\psi\|_{L^2}=1\}\,.
\eeqn
We consider the dyadic decomposition of $\cS=\cup_{j\in\N_0}\cN_j$ based on the sets
\eqn 
    \cN_j &:=& \Big\{\phi\in L^2(\R^3)\,\Big|\,  \|\phi\|_{L^2}=1 \,,\,
    2^{j-1}<\|\phi\|_{\dot H^1}\leq 2^{j}  \Big\}
    \nonumber\\
    \cN_0 &:=& \Big\{\phi\in L^2(\R^3)\,\Big|\,   \|\phi\|_{L^2}=1 \,,\,
    \|\phi\|_{\dot H^1}\leq 1\, \Big\} \,,
\eeqn
where  
$$\|f\|_{\dot H^1}=(\int d\xi |\xi|^{2}|\widehat f(\xi)|^2)^{1/2}\,.$$
We define 
$$d\mu_j(\phi):=d\mu(\phi)\1_{\cN_j}(\phi)\,.$$
Our goal is to introduce subsets $\cM_j\subset\cN_j$, for $j\in\N_0$, such that  for initial data $\phi^{(j)}\in\cM_j$,
the blowup time $\tau(\phi^{(j)})$ for the cubic focusing NLS 
tends to zero as $j\rightarrow\infty$.

For $\phi\in\cN_j$, one observes that if $E[\phi]=\frac12\|\nabla\phi\|_{L^2}^2-\frac14\|\phi\|_{L^4}^4<0$, 
then
\eqn
    \|\phi\|_{L^4}\geq 2^{-\frac14} 2^{\frac j2} \,.
\eeqn
On the other hand, from the Gagliardo-Nirenberg inequality,
\eqn
     \|\phi\|_{L^4}\leq C \|\phi\|_{\dot H^1}^{3/4}\|\phi\|_{L^2}^{1/4} \leq C 2^{\frac{3j}{4}} \,. 
\eeqn
These are the only restrictions on the size of $\|\phi\|_{L^4}$ on $ \cN_j $.

Moreover, from the uncertainty principle
\eqn 
    \|\phi\|_{L^2}^2\leq C \|x\phi\|_{L^2} \|\phi\|_{\dot H^1} \,,
\eeqn
it follows that for $\phi\in\cN_j$,
\eqn
     \|x\phi\|_{L^2} > C 2^{-j} \,.
\eeqn
Thus, we define subsets of $\cN_j$ given by
\begin{align} 
    \cM_j &:= \Big\{\phi\in L^2(\R^3)\,\Big|\, \|\phi\|_{L^2}=1 \,, \|x\phi\|_{L^2} < b \,, \,
    2^{j-1}<\|\phi\|_{\dot H^1}\leq 2^{j} \,, 
    \|\phi\|_{L^4} > C 2^{\frac{5j}{8}}\Big\}
    \nonumber\\
    \cM_0 &:= \Big\{\phi\in L^2(\R^3)\,\Big|\, \|\phi\|_{L^2}=1\,, \|x\phi\|_{L^2} < b \,, \,
    \|\phi\|_{\dot H^1}\leq 1
    \Big\} \,,
\end{align} 
where $b>0$ is a fixed constant.
These sets are non-empty; an example of a function $f_j\in\cM_j$ is given by  
\eqn 
    f_j(x) = 2^{3j/2} g(2^{j}x) \,,
\eeqn
where $g(x)=e^{-x^2}$ is the standard Gaussian.
We define measures $\mu_j$  on $L^2(\R^3)$ satisfying
\eqn\label{eq-mujMj-def-1}
    \mu_j(\cM_j) = \kappa_r (j^{j^{1/\delta} })^{-j}
\eeqn
for $r>1$ and $\delta:=r-1$, where the constant $\kappa_r$ ensures that  $\mu:=\sum \mu_j$ is a probability measure 
on $\cS$. For instance, we can think of $\mu_j$ as the uniform measure concentrated
on $\{e^{i\theta}f_j\}_{\theta\in[0,2\pi)}$, which is invariant under multiplication by a phase.

Then,  we let 
\eqn 
    \gamma^{(k)} := \int d\mu(\phi) (|\phi\rangle\langle\phi|)^{\otimes k} 
    \;\;\; \forall k\in\N \,,
\eeqn
and obtain that 
\eqn 
    \tr(|S^{(k,1)}\gamma^{(k)}|)
    &=&\tr\Big| \, S^{(k,1)}\int d\mu(\phi) (|\phi\rangle\langle\phi|)^{\otimes k} \, \Big|
    \nonumber\\
    &=&\sum_j\int d\mu_j(\phi) \|\phi\|_{H^1}^{2k}
    \nonumber\\
    &\leq&C \sum_j (j^{j^{1/\delta} })^{-j} 2^{2jk}
    \nonumber\\
    &\leq&C e^{c k^r} \,,
    \label{eq-superexp-bd-1}
\eeqn
see Lemma \ref{lm-exponbd-1} below. Thus, $\gamma^{(k)}\in\frH^{1,r}$ for $r>1$.

On the other hand, $(\gamma^{(k)})_{k\in\N}\not\in\frH^{1,1}$.
This is because if $(\gamma^{(k)})_{k\in\N}\in\frH^{1,1}$, 
it follows from Chebyshev's inequality (similar to Lemma \ref{lm-Chebyshev-1}) that
\eqn 
    \mu\Big(\Big\{\phi\in L^2(\R^3)\Big| \|\phi\|_{H^1}>R\Big\}\Big) = 0 \,,
\eeqn
for some $R<\infty$. But this implies that
there are some constants $0<c<C<\infty$ independent of $R$, and $J>0$ 
such that $c\log R<J< C \log R$ for all $R>1$ sufficiently large,
and $\mu(\cM_j)=0$ for all $j>J$.
But then, $\mu_j(\cM_j)=0$ for all $j>J$, which contradicts \eqref{eq-mujMj-def-1}.

For $\phi\in \cM_j$, we have that 
\eqn 
    E[\phi] &=& \frac12 \|\nabla\phi\|_{L^2}^2 - \frac14 \|\phi\|_{L^4}^4
    \nonumber\\
    &<&\frac14(2^{2j}-2^{\frac52 j})
    \nonumber\\
    &<&- C 2^{\frac 52j}
\eeqn
for a constant $C>0$ independent of $j$. Therefore, by the blowup criterion 
of Vlasov-Petrishchev-Talanov \cite{vlpetal}, Zakharov \cite{zakh}, and Glassey \cite{glassey},
the solution $\phi(t)$ with initial data $\phi(0)=\phi$ blows up in finite time in $H^1$.

Next, we derive an upper bound $T_j$ on the blowup time for solutions of the 
focusing cubic NLS with initial data $\phi\in\cM_j$. 
From \eqref{eq-blowup-quadeq-1}, we obtain the quadratic inequality 
\eqn\label{eq-quadrineq-1}
    0&=& \|x\phi(0)\|_{L^2}^2 +
     2 t  \|x\phi\|_{L^2} \|\phi\|_{\dot H^1} +8t^2 E[\phi]
     \nonumber\\
     &\leq&  b^2 + 2 t b - 8 t^2 C 2^{\frac52 j} \,.
\eeqn
The positive zero $T_j>0$ of the quadratic polynomial in $t$ on the lower line
provides an upper bound on the blowup time of the solution $\phi(t)$. 
From \eqref{eq-quadrineq-1}, we get
\eqn 
    T_j < C 2^{-\frac{5j}2}
\eeqn
for a positive constant $C$ independent of $j$. 

Hence, for any $\e>0$, there exists $J=J(\e)>c |\log \e|>0$ such that 
\eqn 
      \tr\Big(\Big|S^{(k,1)}\Big[
      \sum_{j=0}^J \int d\mu_j(\phi) (|S_t(\phi) \rangle
      \langle S_t(\phi)|)^{\otimes k}\Big]
      \Big|\Big)
\eeqn
blows up in a time interval $[0,2^{-cJ})\subset[0,\e)$.
Letting $\e\rightarrow0$ so that $J\rightarrow\infty$, we obtain that 
$\tr(|S^{(k,1)}[\gamma^{(k)}(t)] |)$ blows up instantaneously.

This completes the proof of Theorem \ref{thm-main-2}.
\qed

Finally, we prove the last step in  \eqref{eq-superexp-bd-1}.
 
\begin{lemma}\label{lm-exponbd-1}
Assume that $r>1$, and let $\delta:=r-1$. Then, for $k\in\N$ sufficently large (depending
only on $\delta$),
\eqn 
    \sum_j (j^{j^{1/\delta} })^{-j} 2^{2jk} 
    \leq e^{c k^r}
\eeqn
for a finite constant $c>0$.
\end{lemma}

\begin{proof}
Clearly, 
\eqn 
    \sum_j (j^{j^{1/\delta} })^{-j} 2^{2jk} =
    \sum_j \Big( \frac{2^{2k} }{j^{j^{1/\delta} }}\Big)^j \,.
\eeqn
Let $J=J(k) = k^\delta$. Then,  
$$\frac{2^{2k} }{j^{j^{1/\delta} }} < \frac{2^{2k}}{k^{\delta k}}< \frac12$$ 
for all $j>J$, 
if $k$ is large enough (depending only on $\delta$).
Therefore, 
\eqn
    \sum_{j>J} \Big( \frac{2^{2k} }{j^{j^{1/\delta} }}\Big)^j 
    <  \sum_{j>J} \Big( \frac12\Big)^j < 1 \,,
\eeqn
for $k$ sufficiently large.
On the other hand, 
\eqn
    \sum_{0\leq j\leq J} \Big( \frac{2^{2k} }{j^{j^{1/\delta}}}\Big)^j 
    \leq  \sum_{0\leq j\leq J} 2^{2kj} 
    \leq J 2^{2kJ} =  k^\delta 2^{2k^{1+\delta}} \leq e^{c k^r} \,,
\eeqn
for a suitable constant $c>0$, as claimed.
\end{proof}

\subsection*{Acknowledgements}  
We are grateful to the anonymous referees for their very useful comments.
The work of T.C. was supported by NSF grants DMS-1009448
and DMS-1151414 (CAREER). 
The work of N.P. was supported by NSF grant DMS-1101192.
The work of R.S. was supported by NSERC.

\end{document}